\documentclass[11pt]{article}

\usepackage{bm}
\usepackage{epsf,amsfonts,amsmath}
\usepackage{amsmath}
\usepackage{amssymb}
\usepackage{amsthm}
\newtheorem{thm}{Theorem}[subsection]
\usepackage{hyperref}
\usepackage[dvips]{color}
\usepackage{color}
\usepackage{graphicx}
\usepackage{authblk}
\usepackage{url}
\newcommand{\field}[1]{\mathbb{#1}}
\newtheorem{proposition}[thm]{Proposition}

\DeclareMathOperator{\tr}{tr}

\DeclareMathOperator{\End}{End}
\DeclareMathOperator{\str}{str}

\title{\textbf{Anticommutative extension of the Adler map}}
\author[1,2]{S. Konstantinou-Rizos\thanks{skonstantin84@gmail.com, skonstantin@chesu.ru}}
\author[3]{A.V. Mikhailov\thanks{A.V.Mikhailov@leeds.ac.uk}}
\affil[1]{Institute of Mathematical Physics and Seismodynamics, Chechen State 
University, Russia}
\affil[2]{Faculty of Mathematics and Computer Technology, Chechen State 
University, Russia}
\affil[3]{Department of Applied Mathematics, University of Leeds, UK}

\begin{document}

\maketitle

\begin{abstract}
We construct  a noncommutative (Grassmann) extension of the well-known Adler Yang-Baxter map. It satisfies the Yang-Baxter equation, it is reversible and birational. Our extension preserves all the properties of the original map except the involutivity.
\end{abstract}

\section{Introduction}
 The study of the set-theoretical solutions of the Yang-Baxter equation, which was formally proposed by Drinfeld in \cite{Drinfeld}, gained a more algebraic flavour in \cite{Buchstaber}. The term ``Yang-Baxter maps" for such solutions was proposed by Veselov in \cite{Veselov}. Those Yang-Baxter maps which possess Lax representation \cite{Veselov2} are of particular interest, since they are associated to integrable mappings \cite{Veselov, Veselov3}, and they also possess a natural connection with integrable partial differential equations through Darboux transformations  \cite{Sokor-Sasha}. On the other hand, noncommutative extensions of integrable systems have been constantly of interest over the last few decades. Recently, in \cite{Georgi-Sasha}, Grassmann extensions of Darboux transformations were constructed for a class of nonlinear Schr\"odinger equations, together with their associated super (Grassmann-extended) differential-difference and difference-difference systems. Following the way proposed in \cite{Georgi-Sasha}, super differential-difference and difference-difference systems were constructed for a noncommutative Darboux transformation for the supersymmetric KdV equation in \cite{Xue-Levi-Liu}, and a year later for a generalised super KdV system in \cite{Xue-Liu}. Moreover, the study of the extension of the theory of Yang-Baxter maps in noncommutative (Grassmann) setting was initiated in \cite{GKM}, where the first examples of Grassmann extended Yang-Baxter maps were constructed. 


In this paper, we construct a noncommutative (Grassmann) extension of the well-celebrated Adler map \cite{Adler}, which, over the last couple of decades, has attracted the interest of many authors in the area of integrable systems. This is due to the fact that, not only it is associated with several concepts of integrability, such as Darboux \& B\"acklund transformations, integrable lattice equations and the theory of Yang-Baxter maps, but also it possesses a quite simple form; thus, it constitutes a convenient map in terms of applications. However, the Adler map lacks the property of being non-involutive, which means that it is a trivial map in terms of its dynamics. The noncommutative extension of the Adler map, which we present in our paper, has equally elegant form and preserves all the properties of the original map; namely, it satisfies the Yang-Baxter equation, it is reversible and birational. Nonetheless, it has the extra significant property of being non-involutive.

The paper is organised as follows. In the next section, we present the basic properties of Grassmann algebras. For more information on Grassmann analysis one can consult \cite{Berezin}. Moreover, we explain some properties of the Grasmmann extended Yang-Baxter maps possessing Lax representation. In section 3, using some recent results on Grassmann extensions of Darboux transformations in the case of a generalised super KdV system \cite{Xue-Liu}, we construct a noncommutative extension of the Adler map. Additionally, we show that, in general, noncommutative extensions of involutive Yang-Baxter maps do not preserve their property of involutivity in their noncommutative version. Finally, in section 4, we close with some concluding remarks.

\section{Preliminaries}
\subsection{Grassmann algebra}
Let $G$ be a $\field{Z}_2$-graded algebra over $\field{C}$ with a unity element $e$. As a linear space $G$  is a direct sum $G=G_0\oplus G_1$ (mod 2), such that $G_iG_j\subseteq G_{i+j}$. Those elements of $G$ that belong either to $G_0$ or to $G_1$ are called \textit{homogeneous}, the ones in $G_1$ are called \textit{odd} (or {\sl fermionic}), while those in $G_0$ are called \textit{even} (or {\sl bosonic}) and $e\in G_0$.

By definition, the parity $|a|$ of an even homogeneous element $a$ is $0$, and it is $1$ for odd homogeneous elements. The parity of the product $|ab|$ of two homogeneous elements is a sum of their parities: $|ab|=|a|+|b|$. Grassmann commutativity means that $ba=(-1)^{|a||b|}ab$ for any homogeneous elements $a$ and $b$. In particular, $\alpha^2=0$, for all odd elements $\alpha\in G_1$, and 
even elements commute with all the elements of $G$.

Lastly, for a square matrix, $M$, of the following block-form
\begin{equation*}
M=\left(
\begin{matrix}
 P & \Pi \\
 \Lambda & L
\end{matrix}\right),
\end{equation*}
where $P$ and $L$ are matrices with even entries, while $\Pi$ and $\Lambda$ possess only odd entries (the block matrices are not necessarily square matrices), we can define the \textit{supertrace} --and denote it by $\str(M)$-- to be the following quantity
\begin{equation*}
\str(M)=\tr (P)-\tr (L),
\end{equation*}
where  $\tr(.)$ is the usual trace of a matrix. We shall use the supertrace later on to generate invariants of the Grassmann extended Adler map.

\subsection{Yang-Baxter map and its Lax reperesentation in the Grassmann case}
Let $V_G=\{(a,\alpha)\,|\, a\in G_0,\  \alpha\in G_1\}$. The definitions of Yang-Baxter maps in the Grassmann and commutative cases formally coincide; the only difference is the set of the objects of the maps. In particular, we say that  a map $S\in\End(V_G\times V_G)$
\begin{equation}\label{Smap}
((x,\chi),(y,\psi))\stackrel{S}{\mapsto}\left((u,\xi),(v,\eta)\right),
\end{equation}
is a \textit{Grassmann extended} Yang-Baxter map if it satisfies the 
\textit{Yang-Baxter equation}:
\begin{equation*}
S^{12}\circ S^{13} \circ S^{23}=S^{23}\circ S^{13} \circ S^{12}.
\end{equation*}
Here $S^{ij}\in \End(V_G\times V_G\times V_G)$, $i,j=1,2,3$, $i\neq j$, are 
defined by the following relations
\begin{equation*}
S^{12}=S\times id, \quad S^{23}=id\times S \quad \text{and} \quad S^{13}=\pi^{12} S^{23} \pi^{12},
\end{equation*}
where $\pi^{12}$ is the involution defined by $\pi^{12}((x,\chi),(y,\psi),(z,\zeta))=((y,\psi),(x,\chi),(z,\zeta))$. Moreover, map (\ref{Smap}) is called \textit{reversible} if $S^{21}\circ 
S^{11}=id$, where $S^{21}$ follows from the definition of $S^{12}$ if we change $(u,\xi)\leftrightarrow (v,\eta)$ and $(x,\chi)\leftrightarrow (y,\psi)$.

Furthermore, we shall be using the term \textit{Grassmann extended parametric Yang-Baxter map} if two parameters $a,b\in G_0$ are involved in the definition of (\ref{Smap}), namely we have a map
\begin{equation}\label{SPmap}
S_{a,b}: ((x,\chi),(y,\psi))\mapsto\left((u,\xi),(v,\eta)\right),
\end{equation}
satisfying the \textit{parametric Yang-Baxter equation}
\begin{equation}\label{SYB_eq1}
S^{12}_{a,b}\circ S^{13}_{a,c} \circ S^{23}_{b,c}=S^{23}_{b,c}\circ S^{13}_{a,c} \circ S^{12}_{a,b}.
\end{equation}

Similar to \cite{Veselov2}, we can consider a problem of refactorisation 
\begin{equation}\label{eqLax}
L_a(u,\xi;\lambda)L_b(v,\eta;\lambda)=L_b(y,\psi;\lambda)L_a(x,\chi;\lambda)
\end{equation}
where $L_a(x,\chi;\lambda)$ is a matrix whose entries are Grassmann-valued, $x,a\in G_0,\chi\in G_1$ and $\lambda$ is a complex variable (a {\em spectral parameter}). The problem can be formulated as following: for a given set $\{a,b,x,y\in G_0,\psi,\chi\in G_1\}$ find such $\{u,v\in G_0,\xi,\eta\in G_1\}$ that equation (\ref{eqLax}) is satisfied identically in $\lambda$. If this problem of refactorisation has a unique solution, then it defines a map (\ref{SPmap}) which satisfies the YB equation (\ref{SYB_eq1}) and it is reversible. The proof of the latter statement is exactly the same as in the commutative case \cite{Veselov3}.

The  invariants of a Grassmann-valued Yang-Baxter map can be found using the supertrace. Applying the supertrace to equation (\ref{eqLax}), we obtain that $\str(L_b(y,\psi)L_a(x,\chi))$ is a generating function of the invariants associated Yang-Baxter map $S_{a,b}$.

\section{Grassmann extension of the Adler map}
A generalised super KdV system \begin{eqnarray}\label{superKdV}
u_t&=&u_{xxx}-6u u_x+6\xi_{xx}\eta+6\eta_{xx}\xi,\nonumber\\
\xi_t &=& 4\xi_{xxx}-6u\xi_x-3u_x\xi, \\
\eta_t&=&4\eta_{xxx}-6u\eta_x-3u_x\eta,\nonumber
\end{eqnarray}
where $u=u(x,t)$ is an even variable, whereas $\xi=\xi(x,t)$ and $\eta=\eta(x,t)$ are odd, was presented and studied in \cite{HolodPak}. The  Darboux-B\"acklund transformations and corresponding noncommutative discrete integrable systems associated with (\ref{superKdV}) were recently constructed in \cite{Xue-Liu}. A Lax matrix-operator associated with (\ref{superKdV}) is given by the following
\begin{equation}\label{x-Lax}
\mathcal{L}=D_x-\lambda \left(
\begin{matrix}
 0 & 0 & 0 \\
 e & 0 & 0 \\
0 & 0 & 0 
\end{matrix}\right)-
\left(
\begin{matrix}
 0 & e & 0 \\
 u & 0 & \xi \\ 
\eta & 0 & 0 
\end{matrix}\right),
\end{equation}
which is the spatial part of the Lax pair for (\ref{superKdV}). We omit the temporal part of the Lax pair, since we shall not be using it in this text.

A Grassmann extended Darboux transformation associated to the Lax operator (\ref{x-Lax}) is represented by the following matrix (see \cite{Xue-Liu})
\begin{equation}\label{Darboux-superKdV}
W=\lambda \left(\begin{matrix} 0 & 0 & 0 \\ e & 0 & 0 \\ 0 & 0 & 0 \end{matrix}\right)+\left(\begin{matrix}\frac{v_1}{2} & e & 0\\ \frac{v_1^2}{4}-p_1+\xi\eta_{\left[1\right]} & \frac{v_1}{2} & \xi \\ \eta_{\left[1\right]} & 0 & e\end{matrix}\right),
\end{equation}
where $v_1$ is an even variable, $p_1\in\field{C}$ and $\xi, \eta_{\left[1\right]}$ are odd. In what follows, we use matrix (\ref{Darboux-superKdV}) to construct a Grassmann extended Yang-Baxter map.

\subsection{Noncommutative Adler map}
According to the Grassmann extended Darboux transformation in (\ref{Darboux-superKdV}), changing $(v_1/2,\xi,\eta_{\left[1\right]};p_1)\rightarrow (x,\chi_1,\chi_2;a)$, we consider the following matrix
\begin{equation}\label{M-Adler-1}
M_a(x,\pmb{\chi})=\lambda M_1 + M_0=\lambda \left(\begin{matrix} 0 & 0 & 0 \\ e & 0 & 0 \\ 0 & 0 & 0 \end{matrix}\right)+\left(\begin{matrix}x & e & 0\\x^2-a+\chi_1\chi_2 & x & \chi_1 \\ \chi_2 & 0 & e\end{matrix}\right),
\end{equation}
where $\pmb{\chi}=(\chi_1,\chi_2)$, and we substitute it into the Lax equation. Then, we have the following.

\begin{proposition}(Noncommutative Adler map)  \label{prop}
Let $x+y$ be an invertible element of the Grassmann algebra. Then, the matrix refactorization problem
\begin{equation}\label{LaxEq-Adler}
M_a(u,\pmb{\xi})M_b(v,\pmb{\eta})=M_b(y,\pmb{\psi})M_a(x,\pmb{\chi}),
\end{equation}
where $M_a=M_a(x,\pmb{\chi})$ is given by $(\ref{M-Adler-1})$ is equivalent to a map 
\begin{equation}\label{mapAdler}
((x,\pmb{\chi}),(y,\pmb{\psi}))\stackrel{S_{a,b}}{\longmapsto} ((u,\pmb{\xi}),(v,\pmb{\eta})),
\end{equation}
given by the following
\begin{subequations}\label{superAdler}
\begin{align}
x\mapsto u&=y+\frac{a-b}{x+y-\chi_1\psi_2};\label{superAdler-a}\\
\chi_1\mapsto \xi_1 &=\psi_1-\frac{a-b}{x+y}\chi_1;\label{superAdler-b}\\
\chi_2\mapsto \xi_2&=\psi_2;\label{superAdler-c}\\
y\mapsto v&=x+\frac{b-a}{x+y-\chi_1\psi_2};\label{superAdler-d}\\
\psi_1\mapsto \eta_1&=\chi_1;\label{superAdler-e}\\
\psi_2\mapsto \eta_2&=\chi_2+\frac{a-b}{x+y}\psi_2,\label{superAdler-f}
\end{align}
\end{subequations}
which satisfies the Yang-Baxter equation (\ref{SYB_eq1}), is a parametric, 
reversible, and birational   map with invariants 
\begin{equation}\label{sinvAdler}
I_1=x+y, \qquad I_0=\chi_1\chi_2+\psi_1\psi_2.
\end{equation}
\end{proposition}

\begin{proof}
Equation (\ref{LaxEq-Adler}) implies that $\xi_2=\psi_2$ and $\eta_1=\chi_1$, and, using the latter, the following system of equations
\begin{subequations}\label{uvxieta}
\begin{align}
u+v&=x+y,\label{uvxieta-a}\\
uv+v^2-b+\chi_1\eta_2&=yx+x^2-a+\chi_1\chi_2,\label{uvxieta-b}\\
u^2-a+\xi_1\psi_2+uv&=y^2-b+\psi_1\psi_2+yx,\label{uvxieta-c}\\
u\chi_1+\xi_1=y\chi_1+\psi_1,& \quad \psi_2v+\eta_2=\psi_2x+\chi_2,\label{uvxieta-d}\\
(u^2-a+\xi_1\psi_2)v+u(v^2-b+\chi_1\eta_2)+\xi_1\eta_2&=(y^2-b+\psi_1\psi_2)x+y(x^2-a+\chi_1\chi_2)+\psi_1\chi_2,\label{uvxieta-e}
\end{align}
\end{subequations}
for variables $u$, $v$, $\xi_1$ and $\eta_2$. Now, from equations (\ref{uvxieta-b}) and (\ref{uvxieta-c}) we obtain
\begin{subequations}\label{eqsu-v}
\begin{align}
v(x+y)-b+\chi_1\eta_2=x(y+x)-a+\chi_1\chi_2,\\
u(x+y)-a+\xi_1\psi_2=y(x+y)-b+\psi_1\psi_2,
\end{align}
\end{subequations}
where we have made use of (\ref{uvxieta-a}). Moreover, expressing $\xi_1$ and $\eta_2$ in terms of $u$ and $v$, using equations (\ref{uvxieta-d}), and substituting to (\ref{eqsu-v}) we find that $u$ and $v$ are given by (\ref{superAdler-a}) and (\ref{superAdler-d}), respectively. Using the latter, from (\ref{uvxieta-d}) follows that
\begin{equation*}
\xi_1=\psi_1-\frac{a-b}{x+y-\chi_1\psi_2}\chi_1,\qquad \eta_2=\chi_2+\frac{a-b}{x+y-\chi_1\psi_2}\psi_2.
\end{equation*}
In the above expressions, we multiply both the numerators and the denominators with the conjugate expression of the latter, namely $x+y+\chi_1\psi_2$, and we use the fact that $\chi_1^2=\psi_2^2=0$. Then, it follows that $\xi_1$ and $\eta_2$ are given by (\ref{superAdler-b}) and (\ref{superAdler-f}), respectively. Lastly, equation  (\ref{uvxieta-e}) is satisfied in view of the rest ((\ref{uvxieta-a})- (\ref{uvxieta-d})).

The supertrace of the quantity $M_b(y,\pmb{\psi})M_a(x,\pmb{\chi})$ reads:
\begin{equation*}
\str(M_b(y,\pmb{\psi})M_a(x,\pmb{\chi}))=2\lambda +(x+y)^2+\chi_1\chi_2+\psi_1\psi_2-a-b-e.
\end{equation*}
Thus, $I=(x+y)^2+\chi_1\chi_2+\psi_1\psi_2=I_1^2+I_0$ is invariant of the map (\ref{mapAdler})-(\ref{superAdler}), where $I_1$ and $I_0$ are given by (\ref{sinvAdler}). However, it can be readily verified that $I_1$ and $I_0$ are invariants themselves. 

From the uniqueness of the refactorisation problem (\ref{LaxEq-Adler}) it follows that the map (\ref{mapAdler})-(\ref{superAdler}) satisfies the Yang-Baxter equation and it is reversible. The birationality of the map follows from the fact that (\ref{LaxEq-Adler}) admits the symmetry 
$(u,\pmb{\xi},v,\pmb{\eta},a,b)\leftrightarrow (y,\pmb{\psi},x,\pmb{\chi},b,a)$.

\end{proof}

Setting  all odd variables in 
(\ref{superAdler}) equal to zero and considering even variables and parameters $a,b$ to be $\field{C}$-valued (the bosonic limit) we obtain the standard Adler  map (\ref{AdlerMap})
\begin{equation}\label{AdlerMap}
(x,y) \stackrel{Y_{a,b}}{\longmapsto} 
\left(y+\frac{a-b}{x+y},x+\frac{b-a}{x+y}\right).
\end{equation}
The Adler map (\ref{AdlerMap}) occurs from the 3-D consistent discrete potential KdV equation \cite{Frank, PNC}. Its Lax representation \cite{Veselov2, Veselov3} occurs from the bosonic limit of (\ref{LaxEq-Adler}); that is, it is given by
\begin{equation*}
L_a(u)L_b(v)=L_b(y)L_a(x),
\end{equation*}
where $L_a=L_a(x)$ is given by $M_a(x,\pmb{\chi})$ in (\ref{M-Adler-1}) if one sets the odd variables equal to zero, namely
\begin{equation*}
L_a(x)=
\left(\begin{matrix}
x & e \\
x^2+a-\lambda & x
\end{matrix}\right).
\end{equation*} 

It can be readily verified that the Adler map (\ref{AdlerMap}) is involutive; yet, for its Grassmann extension (\ref{mapAdler})-(\ref{superAdler}) we have the following.

\begin{proposition}\label{Supersymmetric non-involutivity}
The map $S_{a,b}\in\End(V_G\times V_G)$ given by (\ref{superAdler}) is non-involutive.
\end{proposition}
\begin{proof}
If we denote 
$$(x,\chi_1,\chi_2,y,\psi_1,\psi_2) 
 \stackrel{S_{a,b}}{\mapsto} ( x',\chi_1',\chi_2', 
y',\psi_1',\psi_2')  
\stackrel{S_{a,b}}{\mapsto}(x'',\chi_1'', \chi_2'',
y'',\psi_1'',\psi_2''), 
$$
then it follows from (\ref{superAdler-c}) 
and (\ref{superAdler-f}) that
\begin{equation*}
\chi_2''=\psi_2'=\chi_2+\frac{a-b}{x+y}\psi_2\ne \chi_2
\end{equation*}
which implies that
\begin{equation*}
S_{a,b}\circ S_{a,b} \neq id,
\end{equation*}  
and this completes the proof.
\end{proof}

\section{Conclusions}
We contributed to the direction initiated in \cite{GKM}, by constructing a noncommutative (Grassmann) extension of the famous Adler map, using some recent results on noncommutative Darboux transformations for a generalised super KdV system \cite{Xue-Liu}. We showed that, in contrast to the original map, the Grassmann extended Adler map is not-involutive. This is an unexpected and quite interesting phenomenon, since involutive maps possess trivial dynamics.

In \cite{Xue-Liu}, Darboux transformations for the system (\ref{superKdV}) are constructed in three different cases. The first case corresponds to the Darboux matrix (\ref{Darboux-superKdV}) which we used to construct the noncommutative Adler map (\ref{mapAdler})-(\ref{superAdler}). The second case of Darboux transformation, in the same paper \cite{Xue-Liu}, is associated with the following matrix
\begin{equation}\label{Darboux-superKdV-1}
W_2=\lambda \left(\begin{matrix} 0 & 0 & 0 \\ e & 0 & 0 \\ 0 & 0 & e\end{matrix}\right)+\left(\begin{matrix}\frac{v_1}{2} & e & -\xi_{\left[1\right]}\\ \frac{v_1^2}{4}-p_1 & \frac{v_1}{2} & -\frac{v_1}{2}\xi_{\left[1\right]} \\ -\frac{v_1}{2}\eta & -\eta & -p_1-\xi_{\left[1\right]}\eta\end{matrix}\right).
\end{equation}
Changing in the above 
$(v_1/2,-\eta,-\xi_{\left[1\right]};p_1)\rightarrow(x,\chi_1,\chi_2;a)$, we 
obtain a Lax operator 
\begin{equation*}
N_a(x,\pmb{\chi})=\lambda \left(\begin{matrix} 0 & 0 & 0 \\ e & 0 & 0 \\ 0 & 0 & e \end{matrix}\right)+\left(\begin{matrix}x & e & \chi_2\\x^2-a & x & x\chi_2 \\ x\chi_1 & \chi_1 & \chi_1\chi_2-a\end{matrix}\right),
\end{equation*} 
which  we used for the construction of a Yang-Baxter map (similar to Proposition \ref{prop}). Surprisingly the map obtained in this way proved to be exactly the same as (\ref{mapAdler})-(\ref{superAdler}) and, moreover, the differential-difference equations corresponding to (\ref{Darboux-superKdV}) and (\ref{Darboux-superKdV-1}) obtained in \cite{Xue-Liu} coinside after a relabeling of variables. This suggests that the Darboux transformations (\ref{Darboux-superKdV}) and (\ref{Darboux-superKdV-1}) are equivalent in a certain sense, namely they should be related via a transformation, although not necessarily a trivial one.

\section*{Acknowledgements}
A.V.M. acknowledges support from the Leverhulme Trust.

\end{document}